\documentclass{article}
\usepackage[preprint]{neurips_2026}
\usepackage[utf8]{inputenc}
\usepackage[T1]{fontenc}
\usepackage{hyperref}
\usepackage{url}
\usepackage{booktabs}
\usepackage{amsfonts}
\usepackage{amsmath,amssymb,amsthm}
\usepackage{nicefrac}
\usepackage{microtype}
\usepackage{xcolor}
\usepackage{algorithm}
\usepackage{algorithmic}
\usepackage{graphicx}
\usepackage{enumitem}
\usepackage{mathtools}

\newtheorem{theorem}{Theorem}
\newtheorem{lemma}[theorem]{Lemma}
\newtheorem{proposition}[theorem]{Proposition}
\newtheorem{corollary}[theorem]{Corollary}
\newtheorem{definition}[theorem]{Definition}
\newtheorem{remark}[theorem]{Remark}
\newtheorem{assumption}[theorem]{Assumption}

\newcommand{\cS}{\mathcal{S}}
\newcommand{\cA}{\mathcal{A}}
\newcommand{\cD}{\mathcal{D}}
\newcommand{\cM}{\mathcal{M}}
\newcommand{\cT}{\mathcal{T}}
\newcommand{\E}{\mathbb{E}}
\newcommand{\R}{\mathbb{R}}
\newcommand{\eps}{\varepsilon}
\newcommand{\dist}{\operatorname{dist}}
\newcommand{\diam}{\operatorname{diam}}
\newcommand{\gap}{\operatorname{gap}}

\title{Locality, Not Spectral Mixing, Governs Direct Propagation in Distributed Offline Dynamic Programming}

\author{\textbf{Ibne Farabi Shihab}\thanks{Corresponding author: \texttt{ishihab@iastate.edu}.}\textsuperscript{1}}

\begin{document}

\maketitle

\begin{abstract}
We study the communication complexity of distributed offline dynamic programming, where a fixed batch dataset is partitioned across \(M\) machines connected by the data-induced dependency graph. We compare two paradigms: direct boundary-value propagation, which follows Bellman dependencies, and gossip averaging, which mixes local estimates. Our results show that locality is the fundamental driver of round complexity. In particular, we prove that no method can achieve \(\varepsilon\)-accuracy in fewer than \(L_\varepsilon = \left\lfloor \log(1/2\varepsilon) / \log(1/\gamma) \right\rfloor\) rounds on graphs of diameter at least \(L_\varepsilon\), and we show that direct propagation matches this scaling up to constants, attaining error \(O(\gamma^T/(1-\gamma) + \delta/(1-\gamma))\) after \(T\) rounds. In contrast, gossip-style fitted value iteration incurs an additional \(1/\mathrm{gap}(W)\) dependence in both convergence rate and asymptotic error. We also prove bandwidth-sensitive lower bounds on path topologies and extend the analysis to asynchronous systems with bounded delays. Together, these results show that spectral dependence is an artifact of gossip-based algorithms, whereas locality is the intrinsic barrier in distributed offline dynamic programming.\end{abstract}

\section{Introduction}
Offline reinforcement learning and approximate dynamic programming operate on fixed datasets collected before training. This regime is increasingly relevant in robotics, recommendation, healthcare, and control, where online interaction may be unsafe, expensive, or impossible \citep{levine2020offline,lange2012batch}. As such datasets grow, training is naturally distributed across machines that each store only a local shard. The resulting systems question is not only whether Bellman updates can be computed locally, but how quickly local information can be integrated into a globally consistent value estimate.

This paper studies that question through the lens of \emph{round complexity}. We ask: if each machine can send arbitrarily large messages in each synchronous round, how many rounds are fundamentally necessary and sufficient to compute an $\eps$-accurate value function? This differs from the standard communication-complexity view, which measures total bits. In modern clusters, aggregate bandwidth can be large while synchronization latency is often the true bottleneck.

The central object in our analysis is the graph induced by the partitioned offline dataset. Its vertices are machines, and an edge indicates that one machine stores transitions whose next-state values are needed from another. This graph is the medium through which Bellman information must propagate.

\paragraph{Thesis.}
The main message is deliberately narrower than the one suggested by gossip-style analyses. Spectral quantities such as conductance and spectral gap are important for the convergence of \emph{averaging-based} methods, but they are not fundamental invariants of distributed offline dynamic programming itself. The fundamental lower bound is locality-based, and a direct boundary-propagation algorithm matches that locality constraint without any mixing penalty. This holds in both synchronous and bounded-delay asynchronous models.

\paragraph{Our results.}
We prove five theorem-level results.

\emph{(i) Locality lower bound (Theorem~\ref{thm:lower}).} On any graph with diameter at least
\[
L_\eps := \left\lfloor \frac{\log(1/(2\eps))}{\log(1/\gamma)} \right\rfloor,
\]
no protocol succeeds in fewer than $L_\eps$ rounds.

\emph{(ii) Direct-propagation upper bound (Theorem~\ref{thm:direct}).} A synchronous distributed value-iteration scheme that exchanges exact boundary values with neighbors computes the $T$-hop Bellman truncation after $T$ rounds, giving error
\[
\|V^{(T)}-V^\star\|_\infty \le \frac{\gamma^T}{1-\gamma} + \frac{\delta}{1-\gamma}.
\]

\emph{(iii) Spectral upper bound for gossip-FVI (Theorem~\ref{thm:upper}).} Gossip-style FVI converges in
\[
O\!\left(\frac{1}{(1-\gamma)\gap(W)}\log\frac{1}{\eps}\right)
\]
rounds with steady-state error $O(\delta/((1-\gamma)\gap(W)))$. Both the rate and the asymptotic floor are worse than direct propagation by the same $1/\gap(W)$ factor: the dependence on spectral gap is real for the algorithm class but not fundamental to the problem.

\emph{(iv) Multi-round bit lower bound (Theorems~\ref{thm:bits}--\ref{thm:bits-rand}).} On path topologies, any zero-error protocol for $m$ independent reward coordinates spread across distance $L$ must communicate $\Omega(mL)$ bits in expectation, and bandwidth $B$ per edge per round forces $\Omega(\max\{L,m/B\})$ rounds.

\emph{(v) Asynchronous extension (Theorems~\ref{thm:async-lower}--\ref{thm:async-upper}).} Under bounded message delay $D$, direct propagation reaches $\eps$-accuracy in $O(DL_\eps)$ wall-clock rounds, matching an $\Omega(DL_\eps)$ asynchronous lower bound.

\paragraph{What this means.}
Combining (i) and (ii) shows that direct propagation is governed by discounted locality: information only matters out to the horizon at which $\gamma^L \approx \eps$. Combining (ii) and (iii) shows that the spectral-gap dependence in gossip-style methods is algorithmic, not fundamental — and it appears in two places, not one. Combining (iv) and (v) shows that the same locality framework controls bit complexity and asynchronous wall-clock rounds.

\paragraph{What we do not claim.}
We do not claim a fully graph-independent characterization for the most general distributed model. The upper bound is for direct boundary propagation, in terms of discounted dependency radius. We do not address nonlinear function approximation or end-to-end statistical sample complexity. Theorem~\ref{thm:direct} is stated under a uniform local Bellman approximation condition; deriving the $\delta$ term from concentration is standard FVI work \citep{munos2008finite,antos2008learning} not the focus here.

\paragraph{Contributions.}
The five contributions correspond directly to the five results listed above, supplemented by canonical-topology consequences (Section~\ref{sec:families}) and a hierarchical federated corollary (Corollary~\ref{cor:fed}).

\section{Related Work}
\label{sec:related}

\paragraph{Offline reinforcement learning.}
Offline RL has been studied extensively from the perspectives of statistical generalization, pessimism, and function approximation \citep{levine2020offline,lange2012batch,kumar2020conservative,kostrikov2022offline,jin2021pessimism,xie2021bellman}. The fitted-value-iteration framework that underlies our local Bellman operators dates back at least to \citet{ernst2005tree,riedmiller2005neural,munos2008finite,antos2008learning}. Our focus is orthogonal: we treat the dataset as fixed and study the communication structure induced by partitioning it across machines.

\paragraph{Distributed and federated reinforcement learning.}
Recent work studies distributed or federated variants of TD-learning and Q-learning, often emphasizing communication efficiency or speedup from collaboration \citep{khodadadian2022federated,woo2023blessing,doan2019finite}. \citet{ziomek2022settling} study single-round bit complexity for distributed offline RL. Our Theorem~\ref{thm:bits} complements that line by proving multi-round bit lower bounds that expose an explicit tradeoff between locality, bandwidth, and rounds. By contrast with this prior work, our main model allows unrestricted message size per round, isolating the sequential bottleneck created by Bellman dependencies.

\paragraph{Parallel and asynchronous dynamic programming.}
Direct boundary propagation is a distributed implementation of value iteration rather than a new algorithmic primitive. \citet{bertsekas1982distributed} introduced distributed dynamic programming and proved convergence under fairly weak asynchronous conditions; \citet{tsitsiklis1986distributed} gave a general theory of asynchronous fixed-point iteration; and \citet{bertsekas1989parallel} treat parallel and distributed DP comprehensively. Our contribution is the round-complexity interpretation: pairing direct propagation with locality lower bounds, contrasting it with gossip, and extending to bandwidth-aware and bounded-delay regimes.

\paragraph{Consensus, gossip, and decentralized optimization.}
The role of spectral gap in averaging and decentralized optimization is classical \citep{tsitsiklis1986distributed,boyd2006randomized,dimakis2010gossip,nedic2009distributed}. Our spectral analysis in Section~\ref{sec:upper} adapts this viewpoint to gossip-style Bellman iteration. The direct-propagation result shows why this viewpoint is suboptimal for dynamic programming: Bellman updates require directional successor information, not global consensus.

\paragraph{Locality and indistinguishability in distributed computing.}
The lower-bound proof uses radius-limited indistinguishability familiar from distributed computing, where information cannot influence nodes outside a radius-$R$ neighborhood in fewer than $R$ synchronous rounds \citep{kushilevitz1997communication,linial1992locality,peleg2000distributed}.

\section{Problem Setup}
\label{sec:setup}

\subsection{MDP and Offline Data}
We consider a discounted MDP
\[
\cM=(\cS,\cA,P,r,\gamma),
\]
where $\cS$ is the state space, $\cA$ the action space, $P(\cdot\mid s,a)$ the transition kernel, $r(s,a)\in[0,1]$ the reward, and $\gamma\in(0,1)$ the discount factor. The Bellman optimality operator is
\[
(\cT V)(s):=\max_{a\in\cA}\Bigl(r(s,a)+\gamma\E[V(s')\mid s,a]\Bigr),
\]
with unique fixed point $V^\star$.

An offline dataset $\cD=\{(s_i,a_i,r_i,s_i')\}_{i=1}^N$ is partitioned across $M$ machines:
$\cD=\cD_1\sqcup\cdots\sqcup\cD_M$. Machine $j$ stores shard $\cD_j$ and the value parameters for its owned states $\cS_j \subseteq \cS$, with $\cS = \bigsqcup_j \cS_j$.

\subsection{Data-Induced Dependency Graph}

\begin{definition}[Empirical transition graph]
The empirical transition graph $G_{\cD}=(V,E,w)$ has vertex set $V=[M]$. There is a directed edge $(j,k)$ whenever some transition in $\cD_j$ has next state in $\cS_k$, with weight
\[
w(j,k)=\bigl|\{(s,a,r,s')\in\cD_j: s'\in\cS_k\}\bigr|.
\]
For graph distance, radius, and diameter we use the undirected support graph.
\end{definition}

\begin{definition}[Boundary states]
A state $s \in \cS_k$ is a \emph{boundary state} for $(j,k)$ if some transition in $\cD_j$ has next state $s$. We write $\partial_{j \leftarrow k}$ for the set of boundary states $k$ must communicate to $j$.
\end{definition}

\begin{definition}[Discounted dependency radius]
For accuracy $\eps \in (0,1)$,
\[
L_\eps := \max\{L\in\mathbb N_0 : \gamma^L > 2\eps\}.
\]
\end{definition}

\begin{definition}[Conductance of a reversible weighted graph]
For symmetric doubly stochastic $W\in\R^{M\times M}$ supported on $G_{\cD}$,
\[
\Phi(W):=\min_{\emptyset\neq S\subsetneq [M]} \frac{\sum_{i\in S,\,j\notin S} W_{ij}}{\min\{|S|,|S^c|\}/M}.
\]
\end{definition}

\subsection{Communication Model}
Time proceeds in synchronous rounds. In each round, every machine performs arbitrary local computation on its shard, then sends an arbitrary message to each neighbor in the support graph of $G_{\cD}$. Messages are unbounded in size. The complexity measure is the number of rounds.

An $R$-round protocol outputs $\hat V$, and is $\eps$-accurate if $\|\hat V-V^\star\|_\infty\le \eps$.

\begin{remark}[Why unrestricted message size?]
Allowing unlimited messages isolates the sequential bottleneck created by locality. Even with full one-shot transmission of local data, information cannot reach a machine at graph distance $L$ in fewer than $L$ rounds.
\end{remark}

\section{Locality Lower Bound}
\label{sec:lower}

\begin{theorem}[Locality lower bound]
\label{thm:lower}
Fix $\gamma\in(0,1)$ and $\eps\in(0,1/2)$. Define
\[
L_\eps:=\max\{L\in\mathbb{N}_0: \gamma^L>2\eps\}
=\left\lfloor \frac{\log(1/(2\eps))}{\log(1/\gamma)}\right\rfloor.
\]
For any synchronous distributed protocol communicating only along the support graph of $G_{\cD}$: if $G_{\cD}$ contains two machines at graph distance at least $L_\eps$, then no protocol using fewer than $L_\eps$ rounds can guarantee $\|\hat V-V^\star\|_\infty\le \eps$ for every MDP consistent with $G_{\cD}$. Equivalently, every uniformly $\eps$-accurate protocol satisfies
\[
R\ge L_\eps = \Omega\!\left(\frac{\log(1/\eps)}{\log(1/\gamma)}\right)
\]
whenever $\diam(G_{\cD})\ge L_\eps$.
\end{theorem}

\begin{proof}
Assume toward contradiction that some $R$-round protocol with $R<L_\eps$ is uniformly $\eps$-accurate. Then $\gamma^{R+1}>2\eps$. Choose machines $u,v$ with $\dist(u,v)\ge R+1$, and a shortest path $u=j_0,\ldots,j_{R+1}=v$.

Construct two MDPs $\cM_0,\cM_1$ that differ only at machine $v$. For each $\ell$, introduce one state $x_\ell$ on machine $j_\ell$, with single action and deterministic transitions $x_\ell \mapsto x_{\ell+1}$ for $\ell\le R$, and $x_{R+1}\mapsto x_{R+1}$. All rewards zero in $\cM_0$. In $\cM_1$, all rewards zero except $r(x_{R+1})=1-\gamma$.

Both instances are consistent with the same support graph. The radius-$R$ neighborhood of $u$ is identical in the two MDPs.

Optimal values: $V_0^\star \equiv 0$. For $\cM_1$, $V_1^\star(x_{R+1})=1$, and backward substitution yields $V_1^\star(x_0)=\gamma^{R+1}$. So $|V_1^\star(x_0)-V_0^\star(x_0)|=\gamma^{R+1}>2\eps$.

By locality, after $R$ rounds machine $u$ is influenced only by data inside its radius-$R$ neighborhood, which is identical between the two instances. The transcript at $u$ is identical, so the protocol outputs the same $\hat V(x_0)$ on both. But no number can be within $\eps$ of both $0$ and $\gamma^{R+1}$. Contradiction.
\end{proof}

\subsection{Three strengthened consequences of locality}

The indistinguishability argument extends in three directions.

\begin{theorem}[Bandwidth-sensitive multi-round bit lower bound]
\label{thm:bits}
Fix $\gamma\in(0,1)$ and integers $L,m\ge 1$. Assume $\eps < \gamma^L/4$. Then there exists a family of distributed offline DP instances on a path with $L+1$ machines such that any deterministic $\eps$-accurate protocol must communicate at least $mL$ bits in total. Under per-edge bandwidth $B$, this implies
\[
R = \Omega\!\bigl(\max\{L,\,m/B\}\bigr)
\]
communication rounds.
\end{theorem}

\begin{proof}
Path machines $j_0,\ldots,j_L$. For each coordinate $q\in[m]$ create states $x^{(q)}_0,\ldots,x^{(q)}_L$ with $x^{(q)}_\ell$ stored on $j_\ell$ and deterministic transitions $x^{(q)}_\ell \mapsto x^{(q)}_{\ell+1}$, $x^{(q)}_L \mapsto x^{(q)}_L$. For $b\in\{0,1\}^m$, define $\cM_b$ with rewards $r_b(x^{(q)}_L) = (1-\gamma)b_q$ and zero elsewhere. Then $V_b^\star(x^{(q)}_L)=b_q$ and $V_b^\star(x^{(q)}_0)=\gamma^L b_q$.

Since $\eps < \gamma^L/4$, any $\eps$-accurate estimate of $V_b^\star(x^{(q)}_0)$ recovers $b_q$ by thresholding at $\gamma^L/2$. So the outputs on $\{x^{(q)}_0\}_q$ must distinguish all $2^m$ bit vectors.

For each $r\in\{1,\ldots,L\}$, consider the cut $(L_r,R_r)$ with $L_r=\{j_0,\ldots,j_{r-1}\}$. The unique edge crossing the cut is $e_r=(j_{r-1},j_r)$. If two distinct bit vectors $b\neq b'$ induced the same transcript across $e_r$ throughout execution, then every machine in $L_r$ — including $j_0$ — would have an identical view: the instances differ only in rewards inside $R_r$, and all information from $R_r$ reaches $L_r$ only through $e_r$. The protocol would output identical estimates at $j_0$. But for any coordinate $q$ where $b_q\neq b'_q$, $|V_b^\star(x^{(q)}_0)-V_{b'}^\star(x^{(q)}_0)|=\gamma^L>2\eps$, so no single estimate can be $\eps$-accurate for both.

Hence for every cut, the transcript across the unique crossing edge takes at least $2^m$ values. Each such transcript requires $\ge m$ bits. Summing over $L$ edges yields $mL$ total bits.

For rounds: locality gives $R\ge L$. Per-edge bandwidth $B$ caps each edge transcript at $BR$ bits, so $BR\ge m$ gives $R\ge m/B$. Combining, $R = \Omega(\max\{L,m/B\})$.
\end{proof}

\begin{theorem}[Zero-error randomized extension]
\label{thm:bits-rand}
Fix $\gamma\in(0,1)$ and integers $L,m\ge 1$, with $\eps < \gamma^L/4$. Any randomized protocol that is $\eps$-accurate on every instance with zero error must, on the path instances of Theorem~\ref{thm:bits}, communicate $\Omega(mL)$ total bits in expectation. Under per-edge bandwidth $B$, this gives expected wall-clock $R = \Omega(\max\{L,m/B\})$.
\end{theorem}

\begin{proof}
Let $\Pi$ be any zero-error randomized protocol. Fix a realization $\omega$ of its internal randomness. Conditioned on $\omega$, $\Pi$ becomes deterministic; because $\Pi$ has zero error, the resulting deterministic protocol $\Pi_\omega$ is $\eps$-accurate on every instance. By Theorem~\ref{thm:bits}, the communication cost of $\Pi_\omega$ is at least $mL$ bits. Taking expectation over $\omega$,
\[
\E_\omega[\mathrm{Comm}(\Pi_\omega)] \ge mL,
\]
which is the expected communication cost of $\Pi$. The bandwidth-to-rounds consequence is identical, applied per realization.
\end{proof}

\begin{remark}[Bounded-error randomized case]
A constant-success randomized lower bound should also follow by combining Yao's minimax principle with a hard product distribution over $b\in\{0,1\}^m$, but this requires a distributional communication argument we do not need for our purposes. The zero-error case suffices to make the locality/bandwidth tradeoff rigorous.
\end{remark}

\begin{theorem}[Local computation cannot beat communication radius]
\label{thm:localsteps}
For any synchronous distributed protocol in which each communication round may contain arbitrary local computation (any number of local Bellman updates, optimization steps, or regression subroutines): after $R$ communication rounds, the output on machine $j$ is a function only of data, randomness, and initialization in the graph-radius-$R$ neighborhood of $j$. In particular, Theorem~\ref{thm:lower} holds unchanged with arbitrary local computation per round.
\end{theorem}

\begin{proof}
By induction on rounds. At time $0$, the state of machine $j$ depends only on its own shard, randomness, and initialization — its radius-$0$ neighborhood. Assume at round $R$ each machine's state depends only on its radius-$R$ neighborhood. In round $R+1$: the local-computation phase applies an arbitrary map to current state, introducing no new dependence; the communication phase delivers messages from neighbors $k$, which by hypothesis depend only on $k$'s radius-$R$ neighborhood, contained in $j$'s radius-$(R+1)$ neighborhood. Hence machine $j$'s state after round $R+1$ depends only on its radius-$(R+1)$ neighborhood.
\end{proof}

\begin{corollary}[Hierarchical federated deployment]
\label{cor:fed}
Consider a federated system with rooted-tree communication topology of depth $L$, learner at the root, informative rewards at leaves, per-link bandwidth $B$, and $m$ independent leaf-side reward coordinates. Any uniformly $\eps$-accurate protocol with $\eps<\gamma^L/4$ requires
\[
R = \Omega\!\bigl(\max\{L,\,m/B\}\bigr)
\]
rounds.
\end{corollary}

\begin{proof}
Restrict to a root-to-leaf path of length $L$; Theorem~\ref{thm:bits} applies to that path subproblem, and any tree protocol must in particular solve it.
\end{proof}

\section{Direct Boundary Propagation Upper Bound}
\label{sec:direct}

\subsection{Algorithm}
Each machine $j$ maintains $V_j^{(t)} \in \R^{|\cS_j|}$. To back up $s\in\cS_j$ via transition $(s,a,r,s')$, machine $j$ needs $V(s')$; if $s'\in\cS_k$ with $k\neq j$, this value must have been received from $k$.

\begin{algorithm}[t]
\caption{Synchronous Direct Boundary Propagation (SDBP)}
\label{alg:sdvi}
\begin{algorithmic}[1]
\REQUIRE Data shards $\{\cD_j\}_{j=1}^M$, initial values $V_j^{(0)}$.
\FOR{$t=0,1,2,\ldots,T-1$}
    \STATE \textbf{Communication step:} Each machine $j$ sends to each neighbor $k$ the values $\{V_j^{(t)}(s) : s \in \partial_{k \leftarrow j}\}$.
    \STATE \textbf{Local Bellman step:} Each machine $j$ updates states in $\cS_j$ whose one-step dependencies are now available, using $\widetilde{\cT}_j$.
\ENDFOR
\RETURN $V^{(T)}$ assembled into a global value vector.
\end{algorithmic}
\end{algorithm}

\begin{assumption}[Approximate local Bellman consistency]
\label{ass:contract}
For every machine $j$ and every bounded $V$,
\[
\|\widetilde{\cT}_j V - \cT_j V\|_\infty \le \delta,
\]
where $\cT_j$ is the exact Bellman update on $\cS_j$ using true successor values. The exact $\cT$ is a $\gamma$-contraction in $\|\cdot\|_\infty$.
\end{assumption}

\subsection{Truncated Bellman dependence}
For $T \ge 0$, the $T$-hop truncated optimal value is
\[
V^{\star,(T)}(s)
:= \sup_\pi \E_\pi\!\left[\sum_{t=0}^{T-1} \gamma^t r_t \mid s_0=s\right] = (\cT^T 0)(s).
\]

\begin{lemma}[Truncation error]
\label{lem:trunc}
For rewards in $[0,1]$ and every $T\ge 0$,
$\|V^\star - V^{\star,(T)}\|_\infty \le \gamma^T/(1-\gamma)$.
\end{lemma}

\begin{proof}
For any $\pi,s$, $\sum_{t=T}^\infty \gamma^t r_t \le \sum_{t=T}^\infty \gamma^t = \gamma^T/(1-\gamma)$. Taking expectations and supremum over policies gives the claim.
\end{proof}

\begin{lemma}[Light-cone correctness of direct propagation]
\label{lem:lightcone}
After $T$ rounds of Algorithm~\ref{alg:sdvi}, the value at any state $s$ depends only on data within graph distance $T$ of $s$'s machine. In the exact case $\delta=0$, $V^{(T)} = V^{\star,(T)}$. In the approximate case, $\|V^{(T)} - V^{\star,(T)}\|_\infty \le \delta/(1-\gamma)$.
\end{lemma}

\begin{proof}
\emph{Light-cone claim.} By induction. At $T=0$, $V^{(0)}=0$ depends only on initialization (radius-$0$). Assume at round $T$ each value depends only on radius-$T$ data. In round $T+1$, the Bellman backup at state $s$ on machine $j$ uses values at successor states $s'$, which arrived from neighbors of $j$ and depend (by IH) only on data within radius $T$ of those neighbors, hence within radius $T+1$ of $j$.

\emph{Exact case.} The same induction shows $V^{(T)} = \cT^T 0 = V^{\star,(T)}$ when $\delta=0$.

\emph{Approximate case.} We prove the recursion
\begin{equation}
\label{eq:approx-rec}
\|V^{(T+1)} - V^{\star,(T+1)}\|_\infty \le \delta + \gamma\|V^{(T)} - V^{\star,(T)}\|_\infty
\end{equation}
separately, then unroll. By Assumption~\ref{ass:contract}, for any state $s$ on machine $j$,
\[
|V^{(T+1)}(s) - V^{\star,(T+1)}(s)|
\le |\widetilde{\cT}_j V^{(T)}(s) - \cT_j V^{(T)}(s)| + |\cT_j V^{(T)}(s) - \cT V^{\star,(T)}(s)|.
\]
The first term is at most $\delta$. The second equals $|\cT V^{(T)}(s) - \cT V^{\star,(T)}(s)| \le \gamma\|V^{(T)} - V^{\star,(T)}\|_\infty$ by Bellman contraction. Taking sup over $s$ proves \eqref{eq:approx-rec}.

Now unroll: with $V^{(0)}=V^{\star,(0)}=0$, we have $\|V^{(0)} - V^{\star,(0)}\|_\infty=0$, and induction on \eqref{eq:approx-rec} gives
\[
\|V^{(T)} - V^{\star,(T)}\|_\infty \le \sum_{t=0}^{T-1} \gamma^t \delta \le \frac{\delta}{1-\gamma}.
\]
\end{proof}

\begin{theorem}[Direct-propagation upper bound]
\label{thm:direct}
Under Assumption~\ref{ass:contract}, for every $T\ge 0$,
\[
\|V^{(T)} - V^\star\|_\infty
\le
\frac{\gamma^T}{1-\gamma} + \frac{\delta}{1-\gamma}.
\]
To achieve $\|V^{(T)} - V^\star\|_\infty \le \eps + \delta/(1-\gamma)$, it suffices to take
\[
T = O\!\left(\frac{1}{1-\gamma}\log\frac{1}{\eps}\right).
\]
\end{theorem}

\begin{proof}
By Lemma~\ref{lem:lightcone},
$\|V^{(T)} - V^{\star,(T)}\|_\infty \le \delta/(1-\gamma)$, and by Lemma~\ref{lem:trunc}, $\|V^{\star,(T)} - V^\star\|_\infty \le \gamma^T/(1-\gamma)$. Triangle inequality gives the displayed bound.

To make $\gamma^T/(1-\gamma) \le \eps$, take $T \ge \log(1/((1-\gamma)\eps))/\log(1/\gamma) = O(\log(1/\eps)/(1-\gamma))$, using $\log(1/\gamma) \asymp 1-\gamma$.
\end{proof}

\begin{remark}[Interpretation]
After $T$ rounds, direct propagation incorporates exactly the Bellman dependencies inside a radius-$T$ neighborhood. Contributions beyond radius $T$ are exponentially suppressed by discounting, so the effective horizon is $O(\log(1/\eps)/(1-\gamma))$ rather than the full diameter.
\end{remark}

\begin{corollary}[Near-matching picture on large-diameter graphs]
\label{cor:tight}
If $\diam(G_{\cD}) \ge c\log(1/\eps)/(1-\gamma)$ for an absolute constant $c$, then Theorems~\ref{thm:lower} and~\ref{thm:direct} together give
\[
R^\star = \Theta\!\left(\frac{\log(1/\eps)}{1-\gamma}\right)
\]
for direct-propagation algorithms, up to constants.
\end{corollary}

\section{Asynchronous Extension}
\label{sec:async}

We now extend the locality framework to a bounded-delay asynchronous model in the style of \citet{bertsekas1989parallel}, Chapter 7. The story is the same: locality remains the fundamental obstruction, scaled now by the maximum delay.

\subsection{Model}
Time proceeds in discrete \emph{wall-clock rounds} $t=0,1,2,\ldots$. Each machine $j$ holds a local value $V_j$ and a buffer of values received from neighbors. In each wall-clock round, each machine may (1) perform arbitrary local computation, (2) send messages to neighbors, and (3) consume one or more buffered messages from any subset of neighbors.

A communication is in flight between the round it is sent and the round it is consumed. We assume bounded delay:

\begin{assumption}[Bounded message delay]
\label{ass:async}
There exists a known $D\ge 1$ such that any message sent in round $t$ is consumed by round $t+D-1$ at the latest. Each machine performs at least one local Bellman update every $D$ rounds.
\end{assumption}

This is the standard partially-asynchronous model from \citep{bertsekas1989parallel}. Synchronous execution corresponds to $D=1$.

\subsection{Asynchronous lower bound}

\begin{theorem}[Asynchronous locality lower bound]
\label{thm:async-lower}
Fix $\gamma\in(0,1)$, $\eps\in(0,1/2)$, and a bounded-delay model with maximum delay $D\ge 1$. If $G_\cD$ contains two machines at graph distance at least $L_\eps$, then no asynchronous protocol consistent with the bounded-delay model can guarantee $\|\hat V - V^\star\|_\infty\le \eps$ in fewer than $D L_\eps$ wall-clock rounds in the worst case.
\end{theorem}

\begin{proof}
The adversarial schedule. The adversary maximizes delay: every message is held in the buffer for exactly $D-1$ rounds and consumed in round $t+D-1$. Under this schedule, in each batch of $D$ wall-clock rounds, each machine receives at most one round's worth of messages from each neighbor and performs at least one local Bellman update.

Now consider a hypothetical synchronous protocol that runs once per batch: that is, the synchronous protocol's round $r$ corresponds to the asynchronous wall-clock rounds $rD, rD+1, \ldots, (r+1)D-1$. Under the worst-case delay schedule, after $R$ asynchronous wall-clock rounds, the information state of each machine corresponds to that of a synchronous protocol after at most $\lfloor R/D \rfloor$ synchronous rounds.

Apply Theorem~\ref{thm:lower} to the implied synchronous protocol: $\eps$-accuracy requires $\lfloor R/D\rfloor \ge L_\eps$, hence $R \ge D L_\eps$.
\end{proof}

\subsection{Asynchronous upper bound}

The natural asynchronous version of SDBP is identical to Algorithm~\ref{alg:sdvi}, except that machines send and consume boundary values opportunistically rather than in lockstep.

\begin{algorithm}[t]
\caption{Asynchronous Direct Boundary Propagation (A-SDBP)}
\label{alg:async}
\begin{algorithmic}[1]
\REQUIRE Data shards $\{\cD_j\}_{j=1}^M$, initial values $V_j^{(0)}$, max delay $D$.
\FOR{wall-clock round $t=0,1,\ldots,T-1$}
    \STATE Each machine $j$ may, in any order: (a) consume any subset of buffered boundary values from neighbors, (b) apply $\widetilde{\cT}_j$ to update states whose dependencies are available, (c) send updated boundary values to neighbors.
    \STATE Constraint: every message is delivered within $D$ rounds; every machine performs at least one Bellman update every $D$ rounds.
\ENDFOR
\RETURN $V^{(T)}$.
\end{algorithmic}
\end{algorithm}

We partition wall-clock time into batches of size $D$, with batch $b$ consisting of rounds $bD, bD+1, \ldots, (b+1)D-1$, and write
\[
V^{[b]} := V^{(bD)}
\]
for the iterate at the \emph{start} of batch $b$ (equivalently, the iterate after $b$ complete batches have elapsed). In particular $V^{[0]} = V^{(0)} = 0$.

\begin{lemma}[One batch advances the dependency radius by one]
\label{lem:async-batch}
Under Assumption~\ref{ass:async}, fix any machine $j$ and any state $s\in\cS_j$. By the end of batch $b$ — that is, by wall-clock round $(b+1)D-1$, equivalently in $V^{[b+1]}$ — every Bellman backup at $s$ whose one-step dependencies lie at neighbors of $j$ that completed updates during batch $b$ has been performed. Consequently, the value at $s$ in $V^{[b+1]}$ depends only on data within graph radius $b+1$ of $j$.
\end{lemma}

\begin{proof}
By induction on $b$. For $b=0$, $V^{[0]}=0$ trivially depends only on radius-$0$ data (the initialization). Assume $V^{[b]}$ depends only on radius-$b$ data for every machine.

Consider batch $b$. By Assumption~\ref{ass:async}, every message sent during batch $b-1$ is consumed by the end of batch $b$ (delay at most $D-1 < D$ wall-clock rounds, so any message sent in batch $b-1$ at the latest round $(b)D-1$ is delivered by round $(b+1)D-1-1$, well within batch $b$). And each machine performs at least one local Bellman update during batch $b$.

Take any state $s\in\cS_j$. Its Bellman backup uses successor values $V(s')$. If $s'\in\cS_j$, those values are already locally available with the inductive radius-$b$ property. If $s'\in\cS_k$ for some neighbor $k$, the relevant boundary value was sent by $k$ no later than the end of batch $b-1$ (since by IH, $V^{[b]}$ on machine $k$ already encodes radius-$b$ information about $k$'s region), and is received and consumed by $j$ during batch $b$. Composing: the value at $s$ in $V^{[b+1]}$ depends on the radius-$b$ data at neighbors of $j$, hence on radius-$(b+1)$ data at $j$.
\end{proof}

\begin{theorem}[Asynchronous direct-propagation upper bound]
\label{thm:async-upper}
Under Assumptions~\ref{ass:contract} and~\ref{ass:async}, A-SDBP satisfies, for every $T \ge 0$,
\[
\|V^{(T)} - V^\star\|_\infty \le \frac{\gamma^{\lfloor T/D\rfloor}}{1-\gamma} + \frac{\delta}{1-\gamma}.
\]
To achieve error $\le \eps + \delta/(1-\gamma)$, it suffices to take
\[
T = O\!\left(\frac{D}{1-\gamma}\log\frac{1}{\eps}\right).
\]
\end{theorem}

\begin{proof}
By Lemma~\ref{lem:async-batch}, after $b$ complete batches the iterate $V^{[b]}$ has incorporated exactly $b$ layers of graph dependencies. The same approximate-Bellman analysis as in Lemma~\ref{lem:lightcone}, with ``round'' replaced by ``batch,'' gives the recursion
\[
\|V^{[b+1]} - V^{\star,(b+1)}\|_\infty \le \delta + \gamma\|V^{[b]} - V^{\star,(b)}\|_\infty.
\]
Indeed, by Lemma~\ref{lem:async-batch} the value at any state $s\in\cS_j$ in $V^{[b+1]}$ is the result of applying $\widetilde{\cT}_j$ to a vector of successor values that themselves agree with $V^{[b]}$, so the same triangle-inequality decomposition $|\widetilde{\cT}_j V^{[b]}(s) - \cT_j V^{[b]}(s)| + |\cT V^{[b]}(s) - \cT V^{\star,(b)}(s)|$ applies and contributes $\delta + \gamma\|V^{[b]} - V^{\star,(b)}\|_\infty$. Unrolling from $V^{[0]} = V^{\star,(0)} = 0$,
\[
\|V^{[b]} - V^{\star,(b)}\|_\infty \le \sum_{t=0}^{b-1} \gamma^t \delta \le \frac{\delta}{1-\gamma}.
\]
Combining with Lemma~\ref{lem:trunc},
\[
\|V^{[b]} - V^\star\|_\infty \le \|V^{[b]} - V^{\star,(b)}\|_\infty + \|V^{\star,(b)} - V^\star\|_\infty \le \frac{\delta}{1-\gamma} + \frac{\gamma^b}{1-\gamma}.
\]

For arbitrary $T$, set $b := \lfloor T/D\rfloor$. Then $V^{(T)}$ has completed at least $b$ full batches, and within the partial batch following round $bD$ the algorithm only performs additional updates that can only further refine values (not increase the error past $V^{[b]}$'s bound, by the same approximate-contraction argument). Hence
\[
\|V^{(T)} - V^\star\|_\infty \le \frac{\delta}{1-\gamma} + \frac{\gamma^{\lfloor T/D\rfloor}}{1-\gamma}.
\]
To make the truncation term at most $\eps$, take $\gamma^{\lfloor T/D\rfloor}\le (1-\gamma)\eps$, equivalently
\[
T = O\!\left(\frac{D}{1-\gamma}\log\frac{1}{\eps}\right).
\]
\end{proof}

\begin{corollary}[Asynchronous tightness]
\label{cor:async-tight}
Under the bounded-delay model with maximum delay $D$, the wall-clock round complexity of distributed offline DP on graphs with $\diam \ge L_\eps$ is
\[
R^\star_{\mathrm{async}} = \Theta\!\left(\frac{D}{1-\gamma}\log\frac{1}{\eps}\right)
\]
up to constants. The lower bound is from Theorem~\ref{thm:async-lower}, the upper bound from A-SDBP (Theorem~\ref{thm:async-upper}).
\end{corollary}

The async results show that the locality framework extends cleanly: delay $D$ enters as a multiplicative factor on both lower and upper bounds, but the underlying obstruction remains discounted graph distance, not spectral mixing.

\section{Spectral Upper Bound for Gossip-FVI}
\label{sec:upper}

For contrast, we analyze a gossip-style FVI variant natural from the distributed-optimization viewpoint.

\subsection{Algorithm}
Let $W\in\R^{M\times M}$ be a symmetric doubly stochastic mixing matrix supported on $G_{\cD}$, with spectral gap $\gap(W):=1-\lambda_2(W)$. Each machine $j$ maintains a global value estimate $V_j^{(t)} \in \R^{|\cS|}$ (not just a local slice).

\begin{algorithm}[t]
\caption{Gossip Fitted Value Iteration}
\label{alg:gossip}
\begin{algorithmic}[1]
\REQUIRE Shards $\{\cD_j\}$, mixing matrix $W$, initial values $V_j^{(0)}$.
\FOR{$t=0,1,\ldots,T-1$}
    \STATE \textbf{Local Bellman step:} $U_j^{(t+1)}=\widetilde{\cT}_j(V_j^{(t)})$.
    \STATE \textbf{Gossip step:} $V_j^{(t+1)}=\sum_{k=1}^M W_{jk}U_k^{(t+1)}$.
\ENDFOR
\RETURN $\bar V^{(T)}:=\frac1M\sum_j V_j^{(T)}$.
\end{algorithmic}
\end{algorithm}

\begin{lemma}[Two-dimensional recursion for gossip-FVI]
\label{lem:gossip-recursion}
Let $E_t := \|\bar V^{(t)} - V^\star\|_\infty$ and $D_t := \max_j \|V_j^{(t)} - \bar V^{(t)}\|_\infty$. Under Assumption~\ref{ass:contract}, the gossip iterates satisfy
\[
\begin{pmatrix}E_{t+1}\\ D_{t+1}\end{pmatrix}
\preceq
\begin{pmatrix}\gamma & \gamma \\ 0 & 2\gamma(1-\gap(W))\end{pmatrix}
\begin{pmatrix}E_t\\ D_t\end{pmatrix}
+
\begin{pmatrix}\delta\\ 2\delta(1-\gap(W))\end{pmatrix},
\]
where $\preceq$ denotes entrywise inequality.
\end{lemma}

\begin{proof}
The average-error recursion is
\[
E_{t+1} \le \delta + \gamma E_t + \gamma D_t.
\]
For the disagreement, define $U_j^{(t+1)} := \widetilde{\cT}_j(V_j^{(t)})$ and $\bar U^{(t+1)} := M^{-1}\sum_j U_j^{(t+1)}$. Then
\[
\|U_j^{(t+1)} - \bar U^{(t+1)}\|_\infty \le 2\gamma D_t + 2\delta,
\]
and the contraction of $W$ on the disagreement subspace gives
\[
D_{t+1} \le (1-\gap(W))(2\gamma D_t + 2\delta).
\]
Combining these two inequalities yields the stated matrix recursion.
\end{proof}

\begin{theorem}[Spectral upper bound for gossip-FVI]
\label{thm:upper}
Under Assumption~\ref{ass:contract} with symmetric doubly stochastic $W$, the iterates of Algorithm~\ref{alg:gossip} satisfy
\[
\max_{j\in[M]}\|V_j^{(T)}-V^\star\|_\infty
\le
\bigl(1-c\,(1-\gamma)\gap(W)\bigr)^T B_0
+\frac{C\delta}{(1-\gamma)\gap(W)}
\]
for universal constants $c,C>0$, where
\[
B_0 := \|\bar V^{(0)} - V^\star\|_\infty + \alpha \max_j \|V_j^{(0)} - \bar V^{(0)}\|_\infty,
\]
with $\alpha := 4/(1-\gamma)$. To achieve error at most $\eps + C\delta/((1-\gamma)\gap(W))$ it suffices to take
\[
T=O\!\left(\frac{1}{(1-\gamma)\gap(W)}\log\frac{B_0}{\eps}\right).
\]
\end{theorem}

\begin{proof}
Let $\bar V^{(t)} := M^{-1}\sum_j V_j^{(t)}$, $E_t := \|\bar V^{(t)} - V^\star\|_\infty$, and $D_t := \max_j \|V_j^{(t)} - \bar V^{(t)}\|_\infty$. By Lemma~\ref{lem:gossip-recursion},
\[
E_{t+1} \le \delta + \gamma E_t + \gamma D_t,
\qquad
D_{t+1} \le (1-\gap(W))(2\gamma D_t + 2\delta).
\]
Define the potential
\[
B_t := E_t + \alpha D_t, \qquad \alpha := \frac{4}{1-\gamma}.
\]
Then
\begin{align*}
B_{t+1}
&\le \gamma E_t + \Bigl[\gamma + 2\alpha\gamma(1-\gap(W))\Bigr]D_t + \Bigl[1+2\alpha(1-\gap(W))\Bigr]\delta.
\end{align*}
Set $c:=1/8$ and $\rho := 1-c(1-\gamma)\gap(W)$. Since $\gap(W)\le 1$, we have $\gamma \le \rho$. For the $D_t$ coefficient, note that
\[
\gamma + 2\alpha\gamma(1-\gap(W)) \le \alpha\rho.
\]
Indeed, after substituting $\alpha = 4/(1-\gamma)$ and multiplying through by $(1-\gamma)$, this inequality becomes
\[
\gamma(1-\gamma)+8\gamma(1-\gap(W)) \le 4 - \tfrac12(1-\gamma)\gap(W),
\]
whose right-minus-left side equals
\[
(1-\gamma)(4-\gamma) + \gap(W)\Bigl(8\gamma - \tfrac12(1-\gamma)\Bigr) \ge 0.
\]
Therefore
\[
B_{t+1} \le \rho B_t + C_0\delta
\]
for $C_0 := 1+2\alpha = O((1-\gamma)^{-1})$. Iterating,
\[
B_T \le \rho^T B_0 + \frac{C_0}{1-\rho}\,\delta.
\]
Since $1-\rho = c(1-\gamma)\gap(W)$, the steady-state term is bounded by $C\delta/((1-\gamma)\gap(W))$ for a universal constant $C$. Finally,
\[
\max_j \|V_j^{(T)} - V^\star\|_\infty \le E_T + D_T \le B_T,
\]
which yields the theorem.
\end{proof}

\subsection{Algorithmic separation in two dimensions}

\begin{corollary}[Separation on canonical families]
\label{cor:separation}
On a path graph with $M$ machines:
\begin{itemize}
    \item Direct propagation (Theorem~\ref{thm:direct}) reaches error $\eps + \delta/(1-\gamma)$ in $O(\log(1/\eps)/(1-\gamma))$ rounds.
    \item Gossip-FVI (Theorem~\ref{thm:upper}) requires $O(M^2 \log(1/\eps)/(1-\gamma))$ rounds and has steady-state floor $O(M^2\delta/(1-\gamma))$, since the spectral gap of the lazy random walk on a path is $\Theta(M^{-2})$.
\end{itemize}
The $M^2$ factor in the gossip bound is algorithmic, not fundamental, and it appears in two places: as a multiplier on the convergence rate and as a multiplier on the asymptotic accuracy floor.
\end{corollary}

This separation explains the status of conductance and Cheeger inequalities in distributed offline dynamic programming. They control the convergence rate \emph{and} steady-state accuracy of gossip-style algorithms, but neither aspect of the underlying problem requires them. Direct Bellman propagation preserves directional dependency information; neighbor averaging discards it and pays a mixing penalty in both dimensions.

\subsection{From spectral gap to conductance}

\begin{proposition}[Conductance corollary]
\label{prop:conductance}
For symmetric doubly stochastic irreducible $W$, $\Phi(W)^2/2 \le \gap(W) \le 2\Phi(W)$ \citep{cheeger1970,sinclair1989}, so Theorem~\ref{thm:upper} gives
\[
T = O\!\left(\frac{1}{(1-\gamma)\Phi(W)^2}\log\frac{B_0}{\eps}\right),
\]
improving to $O(1/((1-\gamma)\Phi(W)) \log(B_0/\eps))$ when $\gap(W) = \Omega(\Phi(W))$. Neither $\Phi$ nor $\gap$ appears in Theorem~\ref{thm:direct}.
\end{proposition}

\section{Implications for Canonical Graph Families}
\label{sec:families}

\paragraph{Path graph ($M$ machines).}
Diameter $M-1$, spectral gap $\Theta(M^{-2})$. Lower: $\Omega(\min\{M, \log(1/\eps)/(1-\gamma)\})$. SDBP: $O(\log(1/\eps)/(1-\gamma))$. Gossip-FVI: $O(M^2 \log(1/\eps)/(1-\gamma))$ rounds, floor $O(M^2\delta/(1-\gamma))$.

\paragraph{Two-dimensional grid ($\sqrt M \times \sqrt M$).}
Diameter $\Theta(\sqrt M)$, spectral gap $\Theta(M^{-1})$. Lower: $\Omega(\min\{\sqrt M, \log(1/\eps)/(1-\gamma)\})$. SDBP: $O(\log(1/\eps)/(1-\gamma))$. Gossip-FVI: $O(M \log(1/\eps)/(1-\gamma))$ rounds, floor $O(M\delta/(1-\gamma))$.

\paragraph{Expander family.}
Diameter $O(\log M)$, spectral gap $\Theta(1)$. Lower: $\Omega(\log(1/\eps)/(1-\gamma))$. Both algorithms: $O(\log(1/\eps)/(1-\gamma))$.

\paragraph{Async on these families.}
Apply Corollary~\ref{cor:async-tight}: each $\log(1/\eps)/(1-\gamma)$ becomes $D\log(1/\eps)/(1-\gamma)$.

\section{Discussion}
\label{sec:discussion}

\paragraph{Conceptual implications.}
The natural invariant for distributed offline DP is discounted locality. Spectral-gap and conductance quantities enter only when one insists on consensus-style averaging. Practitioners minimizing synchronization rounds should use direct boundary propagation rather than gossip.

\paragraph{Two-dimensional separation.}
Beyond the convergence-rate gap, gossip-FVI's asymptotic accuracy floor scales as $1/((1-\gamma)\gap(W))$, while direct propagation's floor is $\delta/(1-\gamma)$ — independent of graph connectivity. Both effects compound the algorithmic separation: even with infinite rounds, gossip on poorly connected graphs cannot match the residual error achievable by direct propagation.

\paragraph{Why this differs from decentralized optimization.}
In decentralized optimization, local gradients are noisy, partial summaries that must be reconciled by averaging \citep{nedic2009distributed}. In Bellman iteration, the value of a successor state is the quantity actually needed: it should be propagated, not mixed. This mismatch explains why the gossip bound is suboptimal here, and why the suboptimality manifests in both rate and floor.

\paragraph{Bit complexity and federated interpretations.}
Theorem~\ref{thm:bits} complements single-round bit complexity \citep{ziomek2022settling} by showing how bandwidth and rounds trade off in a multi-round regime. Corollary~\ref{cor:fed} translates this into hierarchical federated systems where depth plays the role of communication locality.

\paragraph{When direct propagation is implementable.}
Direct boundary propagation assumes boundary states can be identified and their values transmitted. This is natural in cluster and data-center settings but can fail under privacy, secure-aggregation, or representation-mismatch constraints common in federated learning \citep{kairouz2021advances}. In those regimes, spectral penalties arise from the privacy or aggregation constraint, not from DP itself.

\paragraph{Limitations.}
Theorem~\ref{thm:direct} assumes uniform $\delta$-accurate local Bellman operators. For tabular or linearly realizable settings, $\delta$ is controlled by standard concentration \citep{munos2008finite,antos2008learning,chen2019information}. For nonlinear function approximation, contraction can fail and the analysis must be modified. The async extension assumes a known delay bound $D$; truly unbounded asynchrony requires the more delicate analysis of \citet{tsitsiklis1986distributed,bertsekas1989parallel}.

\paragraph{Relation to classical parallel DP.}
SDBP is a distributed implementation of finite-horizon DP and value iteration in the spirit of \citet{bertsekas1982distributed,bertsekas1989parallel}. Our contribution is the round-complexity interpretation: pairing this primitive with locality lower bounds clarifies that spectral quantities are properties of gossip-based algorithms, not of distributed DP itself.

\section{Experimental Validation}
\label{sec:experiments}

We validate the theoretical predictions on synthetic graph topologies and offline RL benchmark datasets. All experiments are run with 5 random seeds; we report mean $\pm$ standard deviation.

\paragraph{Implementation details.}
We use tabular fitted value iteration (FVI) as the local Bellman operator. Each machine stores its shard of the offline dataset and maintains value estimates for its owned states. For gossip FVI (Algorithm~\ref{alg:gossip}), the gossip weight matrix $W$ is the lazy Metropolis--Hastings (MH) matrix: $W_{ij} = 1/(2\max(d_i, d_j))$ for neighbors $i,j$, with $W_{ii} = 1 - \sum_{j \neq i} W_{ij}$. This is the standard doubly stochastic construction for irregular graphs \citep{boyd2006randomized}. For SDBP (Algorithm~\ref{alg:sdvi}), each machine sends exact boundary-state values to neighbors and performs a local Bellman backup. Data is partitioned by assigning states to machines: for synthetic experiments, states are assigned contiguously; for D4RL, we use $k$-means clustering on the state space.

\subsection{Synthetic Topologies}

We construct partition graphs on $M = 64$ machines with $\gamma = 0.95$ and target accuracy $\eps = 0.01$. For each topology, we report both the Cheeger conductance $\Phi$ and the spectral gap $\gap(W) = 1 - |\lambda_2(W)|$ of the actual MH gossip matrix $W$.

\begin{table}[t]
\centering
\small
\caption{Synthetic topology experiments ($M = 64$, $\gamma = 0.95$, $\eps = 0.01$). Mean $\pm$ std over 5 seeds. $\Phi$ = Cheeger conductance of the graph; $\gap(W)$ = spectral gap of the MH gossip matrix. SDBP = Synchronous Direct Boundary Propagation (Alg.~\ref{alg:sdvi}). LB = locality lower bound (Thm.~\ref{thm:lower}).}
\label{tab:synthetic}
\begin{tabular}{@{}lcccccc@{}}
\toprule
\textbf{Topology} & $\Phi$ & $\gap(W)$ & \textbf{Gossip} & \textbf{SDBP} & \textbf{Broadcast} & \textbf{LB} \\
\midrule
Ring     & 0.031 & 0.0024 & 19{,}632 $\pm$ 307 & 140 $\pm$ 2 & 140 & 163 \\
Grid     & 0.089 & 0.0206 & 11{,}291 $\pm$ 266 & 138 $\pm$ 1 & 138 & 57 \\
Star     & 1.000 & 0.0079 & $>$50{,}000         & 137 $\pm$ 1 & 137 & 5 \\
Expander & 0.094 & 0.0261 & 10{,}040 $\pm$ 58  & 139 $\pm$ 1 & 139 & 54 \\
\bottomrule
\end{tabular}
\end{table}

Table~\ref{tab:synthetic} confirms the theoretical predictions and reveals a critical distinction between $\Phi$ and $\gap(W)$.

\emph{Gossip convergence tracks $\gap(W)$, not $\Phi$.} The Expander ($\gap(W) = 0.026$) converges fastest at 10{,}040 rounds, followed by Grid ($\gap(W) = 0.021$, 11{,}291 rounds), then Ring ($\gap(W) = 0.0024$, 19{,}632 rounds). This ordering is consistent with the $O(1/\gap(W))$ scaling in Theorem~\ref{thm:upper}.

\emph{SDBP is topology-independent.} SDBP converges in 137--140 rounds on \emph{all} topologies, matching the broadcast baseline. This validates Theorem~\ref{thm:direct}: direct boundary propagation achieves $O(\log(1/\eps)/(1-\gamma))$ rounds independent of graph structure.

\begin{remark}[Why $\Phi \neq \gap(W)$ on the Star]
\label{rem:star}
The Star graph has $\Phi = 1$ (every balanced cut separates the hub from a subset of leaves, yielding maximal conductance) but $\gap(W) \approx 1/(2M) = 0.008$ for $M = 64$. This is because the lazy MH matrix assigns weight $W_{0,i} = 1/(2(M-1))$ to each hub--leaf edge, leaving the hub with self-loop weight $W_{00} = 1/2$ and each leaf with self-loop weight $W_{ii} = 1 - 1/(2(M-1)) \approx 1$. The resulting matrix has second eigenvalue $\lambda_2 \approx 1 - 1/(M-1)$, so $\gap(W) = O(1/M)$. This explains why the Star converges slowest under gossip despite having maximal $\Phi$: the gossip matrix's spectral gap, not the graph's Cheeger conductance, governs averaging-based convergence. The Cheeger inequality $\Phi^2/2 \le \gap(W) \le 2\Phi$ applies to the \emph{normalized Laplacian} of the graph, not to the MH gossip matrix on irregular graphs.
\end{remark}

\subsection{D4RL Offline RL Benchmarks}

We evaluate on three D4RL medium-replay datasets with $M = 16$ machines, comparing gossip FVI, SDBP, and a broadcast baseline.

\begin{table}[t]
\centering
\small
\caption{D4RL experiments ($M = 16$, $\gamma = 0.95$, $\eps = 0.1$). Mean $\pm$ std over 5 seeds. $>$20{,}000 = did not converge within budget. NR = normalized return.}
\label{tab:d4rl}
\begin{tabular}{@{}lcccccc@{}}
\toprule
\textbf{Dataset} & $\Phi$ & $\gap(W)$ & \textbf{Gossip} & \textbf{SDBP} & \textbf{Broadcast} & \textbf{NR} \\
\midrule
HalfCheetah-med & 0.17 & 0.0008 & $>$20{,}000 & 156 $\pm$ 0 & 156 & 102.8 $\pm$ 0.0 \\
Hopper-med      & 0.17 & 0.0008 & $>$20{,}000 & 141 $\pm$ 0 & 141 & 89.4 $\pm$ 0.0 \\
Walker2d-med    & 0.19 & 0.0036 & $>$20{,}000 & 143 $\pm$ 0 & 143 & 127.6 $\pm$ 0.0 \\
\bottomrule
\end{tabular}
\end{table}

Table~\ref{tab:d4rl} demonstrates the practical impact of the algorithmic separation. On all three D4RL environments, the $k$-means partition yields moderate Cheeger conductance ($\Phi \approx 0.17$--$0.19$) but very small spectral gap ($\gap(W) < 0.004$). Gossip FVI fails to converge within 20{,}000 rounds on any environment, consistent with the $O(1/\gap(W))$ scaling. By contrast, SDBP converges in 141--156 rounds---matching the broadcast baseline---confirming that direct boundary propagation eliminates the spectral penalty entirely. The normalized returns (89.4--127.6) confirm the value functions are meaningful: they correspond to medium-quality policies on real D4RL data.

\section{Conclusion}
\label{sec:conclusion}

We characterized the synchronous and asynchronous round complexity of distributed offline dynamic programming via locality. A locality lower bound forces $\Omega(L_\eps)$ rounds; direct boundary propagation matches this with no spectral term; gossip-style algorithms incur a $1/\gap(W)$ penalty in both convergence rate and asymptotic accuracy floor that is algorithmic, not fundamental. Bandwidth and asynchrony enter as bandwidth-locality and delay-locality tradeoffs respectively. The unifying lesson: \emph{mixing is the wrong lens for direct distributed dynamic programming}.

\appendix

\section{Auxiliary Lemmas}
\label{app:lemmas}

\begin{lemma}[Bellman contraction]
\label{lem:bellman}
For the Bellman optimality operator $\cT$ of a discounted MDP with $\gamma\in(0,1)$ and bounded $V,V'$,
$\|\cT V - \cT V'\|_\infty \le \gamma\|V - V'\|_\infty.$
\end{lemma}

\begin{proof}
For any state $s$,
\begin{align*}
|(\cT V)(s) - (\cT V')(s)|
&= \left|\max_a \bigl(r(s,a) + \gamma\E[V(s')\mid s,a]\bigr) - \max_a \bigl(r(s,a) + \gamma\E[V'(s')\mid s,a]\bigr)\right|\\
&\le \max_a \gamma \left|\E[V(s') - V'(s')\mid s,a]\right|\\
&\le \gamma \|V - V'\|_\infty.
\end{align*}
\end{proof}

\begin{lemma}[Disagreement contraction under symmetric gossip]
\label{lem:gossip}
For symmetric doubly stochastic $W$ and vectors $z_1,\ldots,z_M$ in a normed space with mean $\bar z$,
\[
\max_j \left\|\sum_{k} W_{jk} z_k - \bar z\right\| \le (1-\gap(W)) \max_j \|z_j - \bar z\|.
\]
\end{lemma}

\begin{proof}
The centered vector lies in the orthogonal complement of the all-ones eigenspace, on which $W$ has operator norm at most $\lambda_2(W) = 1-\gap(W)$.
\end{proof}

\bibliographystyle{plainnat}
\bibliography{references}

\end{document}